\documentclass[preprint,12pt,authoryear]{elsarticle}

\usepackage{amssymb}
\usepackage{amsthm}
\usepackage{amsmath}
\usepackage{graphicx}
\usepackage{color}
\usepackage{latexsym}

\usepackage{extpfeil}

\newcommand{\cD}{{\mathcal D}}
\newcommand{\cE}{{\mathcal E}}
\newcommand{\cN}{{\mathcal N}}
\newcommand{\cS}{{\mathcal S}} 
\newcommand{\cT}{{\mathcal T}} 
\newcommand{\ch}{{\rm ch}}

\newtheorem{theorem}{Theorem}
\newtheorem{lemma}{Lemma}
\newtheorem{observation}{Observation}
\newtheorem{corollary}{Corollary}

\journal{XXX}

\begin{document}

\begin{frontmatter}

\title{On the Quirks of Maximum Parsimony and Likelihood on Phylogenetic Networks}

\author[stats]{Christopher Bryant}\ead{chris.bryant@stats.govt.nz}\author[greifswald] {Mareike Fischer}\ead{email@mareikefischer.de} \author[UoA] {Simone Linz}\ead{s.linz@auckland.ac.nz}  \author[UC] {Charles Semple}\ead{charles.semple@canterbury.ac.nz}

\address[stats]{Statistics New Zealand, Wellington, New Zealand.}
\address[greifswald]{Department for Mathematics and Computer Science, Ernst Moritz Arndt University Greifswald, Germany.}
\address[UoA]{Department of Computer Science, University of Auckland, New Zealand.}
\address[UC]{School of Mathematics and Statistics, University of Canterbury, Christchurch, New Zealand.}

\begin{abstract}
Maximum parsimony is one of the most frequently-discussed tree reconstruction methods in phylogenetic estimation. However, in recent years it has become more and more apparent that phylogenetic trees are often not sufficient to describe evolution accurately. For instance, processes like hybridization or lateral gene transfer that are commonplace in many groups of organisms and result in mosaic patterns of relationships cannot be represented by a single phylogenetic tree. This is why phylogenetic networks, which can display such events, are becoming of more and more interest in phylogenetic research. It is therefore necessary to extend concepts like maximum parsimony from  phylogenetic trees to networks. Several suggestions for possible extensions can be found in recent literature, for instance the softwired and the hardwired parsimony concepts. In this paper, we analyze the so-called big parsimony problem under these two concepts, i.e. we investigate maximum parsimonious networks and analyze their properties. In particular, we show that finding a softwired maximum parsimony network is possible in polynomial time. We also show that the set of maximum parsimony networks for the hardwired definition always contains at least one phylogenetic tree. Lastly, we investigate some parallels of parsimony to different likelihood concepts on phylogenetic networks.

\end{abstract}


\end{frontmatter}

\section{Introduction}
Maximum parsimony (MP) is a popular tool to reconstruct phylogenetic trees from a sequence of morphological or molecular characters. Since there is currently an increasing interest in representing evolution as an intertwined network~\citep{bapteste13,morrison11} that accounts for speciation as well as reticulation events such as lateral gene transfer or hybridization, it is not surprising that consideration is given to extending parsimony to phylogenetic networks. Similar to parsimony on phylogenetic trees (reviewed in \citet{felsenstein04}), one distinguishes the small and big parsimony problem. In terms of phylogenetic networks, the small parsimony problem asks for the parsimony score of a sequence of characters on a (given) phylogenetic network, while the big parsimony problem asks to find a phylogenetic network for a sequence of characters that minimizes the score amongst all phylogenetic networks. It is the latter problem that evolutionary biologists usually want to solve for a given data set, and it is this problem that is the focus of this paper.

Recently, two different approaches for parsimony on phylogenetic networks have been proposed, referred to as  {\it hardwired} and {\it softwired} parsimony. The hardwired framework, introduced by \citet{kannan12}, calculates the parsimony score of a phylogenetic network by considering character-state transitions along every edge of the network. A slightly different approach was taken by \citet{nakhleh05}, who defined the softwired parsimony score of a phylogenetic network to be the smallest (ordinary) parsimony score of any phylogenetic tree that is displayed by the network under consideration. Although one can compute the hardwired parsimony score of a set of binary characters on a phylogenetic network in polynomial time \citep{semple03}, solving the small parsimony problem is in general NP-hard under both notions \citep{fischer,jin09,nguyen07}.
In contrast, the small parsimony problem on phylogenetic trees is solvable in polynomial time by applying Fitch-Hartigan's~\citep{fitch71,hartigan73} or Sankoff's \citep{sankoff75} algorithm.

Given that it is in general computationally expensive to solve the small parsimony problem on networks, effort has been put into the development of heuristics \citep{kannan12}, and algorithms that are exact and have a reasonable running time despite the complexity of the underlying problem \citep{fischer,kannan14}. However, in finding ever quicker and more advanced algorithms to solve the small parsimony problem, an analysis of MP networks under the hardwired or softwired notion, and their biological relevance has fallen short. The only exceptions are two practical studies \citep{jin06,jin07} that aim at the reconstruction of a particular type of a softwired MP network for which the input does not only consist of a  sequence of characters, but also of a given phylogenetic tree $\cT$ (e.g. a species tree) and a positive integer $k$. More precisely, this version of softwired parsimony adds $k$ reticulation edges to $\cT$ such that the softwired parsimony score of the resulting phylogenetic network is minimized over all possible solutions.

In this paper, we present the first analysis of MP networks and reveal fundamental properties of such networks that are simultaneously surprising and undesirable. For example, we show that an MP network under the hardwired definition tends to have a small number of reticulations, while an MP network under the softwired definition tends to have many reticulations. Even stronger, we show that, for any sequence of characters, there always exists a phylogenetic tree that is an MP network under the hardwired definition. While some of our findings have independently been stated in \citet{wheeler15}, we remark that the author does not give any formal proofs. In conclusion, the properties we find question the biological meaningfulness of MP networks and emphasize a fundamental difference between the hardwired and softwired parsimony framework on phylogenetic networks. We then shift towards maximum likelihood concepts on phylogenetic networks and analyze whether or not the Tuffley-Steel equivalence result for phylogenetic trees also holds for networks. It is well known that under a simple substitution model, parsimony and likelihood on phylogenetic trees are equivalent \citep{TS}. However, as we shall show, parsimony on networks is not equivalent to one of the most frequently-used likelihood concepts on networks. Nevertheless, the equivalence can be recovered using functions that resemble likelihoods, but are not true likelihoods in a probability theoretical sense. We call these functions pseudo-likelihoods. In this sense, the equivalence of the different parsimony concepts to pseudo-likelihoods rather than likelihoods can be viewed as another drawback of the existing notions of parsimony.

The remainder of the paper is organized as follows. The next section contains notation and terminology that is used throughout the paper. We then analyze properties of MP networks under the hardwired and softwired definition in Section~\ref{sec:parsimony}. Additionally, this section also considers the computational complexity of the big parsimony problem under both definitions. Then, in Section~\ref{sec:MPML}, we re-visit the Tuffley-Steel equivalence result for parsimony and likelihood, and investigate in how far it can be extended from trees to networks. We end the paper with a brief conclusion in Section~\ref{sec:conclu}. 

Lastly, it is worth noting that our results are presented as general as possible. For example, we do not bound the number of character states of any character that is considered in this paper. Furthermore, the only restriction in the definition of a phylogenetic network (see next section for details) is that the out-degree of a reticulation is exactly one. As a reticulation and speciation event are unlikely to happen simultaneously, this restriction is biologically sensible and, in fact, only needed to establish Theorem~\ref{t:newDef}.

\section{Preliminaries}\label{sec:prelim}
\subsection{Trees and networks} 
A {\it rooted phylogenetic tree on $X$} is a rooted tree with no degree-two vertices (except possibly the root which has degree at least two) and whose leaf set is  $X$. Furthermore, a rooted phylogenetic tree on $X$ is {\it binary} if each internal vertex, except for the root, has degree three. A natural extension of a rooted phylogenetic tree on $X$ that allows for vertices whose in-degree is greater than one is a {\it rooted phylogenetic network $\cN$ on $X$} which is a rooted acyclic digraph that satisfies the following three properties:
\begin{itemize}
\item[(i)] $X$ is the set of vertices of in-degree one and out-degree zero,
\item[(ii)] the out-degree of the root is at least two, and 
\item[(iii)] every other vertex has either in-degree one and out-degree at least two, or in-degree at least two and out-degree one.
\end{itemize}
Similar to rooted phylogenetic trees, we call $X$ the {\it leaf set} of $\cN$. Furthermore, each vertex of $\cN$ whose in-degree is at least two is called a {\it reticulation} and represents a species whose genome is a mosaic of at least two distinct parental genomes, while each edge directed into a reticulation is called a {\it reticulation edge}. To illustrate, a rooted phylogenetic network on $X=\{1,2,3,4\}$ and with one reticulation is shown on the left-hand side of Figure~\ref{fig:prelim}. Moreover, for two vertices $u$ and $v$ in $\cN$, we say that $u$ is a {\it parent} of $v$ or, equivalently, $v$ is a {\it child} of $u$ if $(u,v)$ is an edge in $\cN$. Lastly, note that a rooted phylogenetic tree on $X$ is a rooted phylogenetic network on $X$ with no reticulation.

\begin{figure}[t]
\center
\scalebox{1.3}{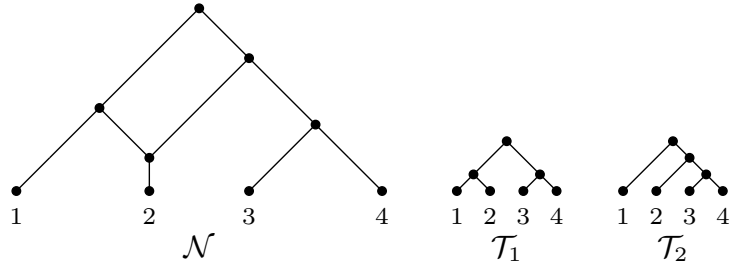}
\caption{Left: A rooted phylogenetic network $\cN$ on leaf set $X=\{1,2,3,4\}$. Right: The two rooted phylogenetic trees $\cT_1$ and $\cT_2$ on $X$ displayed by $\cN$.}
\label{fig:prelim}
\end{figure}

Let $\cN$ be a rooted phylogenetic network on $X$ and let $\cT$ be a rooted phylogenetic tree on $X$. 
We say that $\cT$ is {\it displayed} by $\cN$ if, up to contracting vertices with in-degree one and out-degree one, $\cT$ can be obtained from $\cN$ by deleting edges and non-root vertices, in which case the resulting acyclic digraph is an {\it embedding} of $\cT$ in $\cN$. Intuitively, if $\cT$ is displayed by $\cN$, then all ancestral information inferred by $\cT$ is also inferred by $\cN$. The two rooted phylogenetic trees $\cT_1$ and $\cT_2$ that are displayed by the network shown on the left-hand side of Figure~\ref{fig:prelim} are presented on the right-hand side of the same figure. Lastly, we use $\cD(\cN)$ to denote the set of all rooted phylogenetic trees that are displayed by $\cN$.

\subsection{Characters} 
Let $G$ be an acyclic digraph. We denote the vertex set of $G$ by $V(G)$ and the edge set of $G$ by $E(G)$. Furthermore, we call $X$  a {\it distinguished set} of $G$ if it is a  subset of the vertices of $G$ whose out-degree is zero such that, if $G$ is a rooted phylogenetic network $\cN$ (resp. a rooted phylogenetic tree $\cT$), then $X$ is precisely the leaf set of $\cN$ (resp. $\cT$). A {\it character on $X$} is a function $\chi$ from $X$ into a set $C$ of character states. 

Let $G$ be an acyclic digraph with distinguished set $X$ and let $\chi$ be a character on $X$. An {\it extension} of $\chi$ to $V(G)$ is a function $\bar\chi$ from $V(G)$ to $C$ such that $\bar\chi(\ell)=\chi(\ell)$ for each element $\ell\in X$. For an extension $\bar\chi$ of $\chi$ to $V(G)$, we set $$\ch(\bar\chi,G)=|\{(u,v)\in E(G):\bar\chi(u)\ne\bar\chi(v)\}|,$$ and refer to it as the {\it changing number} of $\bar\chi$. In other words, the changing number of $\bar\chi$ is the number of edges in $G$ whose two endpoints are assigned to different character states. Two characters $\chi_1$ and $\chi_2$ on $X$ are shown on the left-hand side of Figure~\ref{fig:hardwired} while possible extensions $\bar\chi_1$ and $\bar\chi_2$ of $\chi_1$ and $\chi_2$, respectively, to the vertex set of the underlying rooted phylogenetic network $\cN$ on four leaves are shown in the middle and on the right-hand side of the same figure. Note that $\ch(\bar\chi_1,\cN)=\ch(\bar\chi_2,\cN)=2$.
If $G$ is a rooted phylogenetic tree on $X$, we say that  $\chi$ is {\it homoplasy-free} on $G$ if there exists an extension $\bar\chi$ of $\chi$ to $V(G)$ such that, for each character state $c_i\in C$, the subgraph of $G$ induced by $\{v\in V(G): \bar{\chi}(v)=c_i\}$, the subset of vertices assigned to the same character state, is connected. Equivalently, $\chi$ is said to be homoplasy-free on $G$ if there exists an extension $\bar\chi$ of $\chi$ to $V(G)$ such that $\ch(\bar\chi,G)=|C|-1$. Biologically speaking, if $\chi$ is homoplasy-free on a rooted phylogenetic tree, then $\chi$ can be explained without any reverse or convergent character-state transitions. Note that, for each character $\chi$ on $X$, there always exists a rooted phylogenetic tree $\cT$ such that $\chi$ is homoplasy-free on $\cT$, in which case $\cT$ is said to be a {\it perfect phylogeny} for $\chi$. 

Now, let $\cT$ be a perfect phylogeny for a character $\chi$ on $X$. It is easily checked that any rooted binary phylogenetic tree $\cT'$ on $X$ with the property that $\cT$ can be obtained from $\cT'$ by contracting a possibly empty set of edges is also a perfect phylogeny for $\chi$. We call $\cT'$ a {\it binary refinement} of $\cT$. The next observation is an immediate consequence of the fact that each phylogenetic tree has a binary refinement.

\begin{observation}\label{ob-pp}
Let $\chi$ be a character on $X$. There exists a rooted binary phylogenetic tree $\cT$ on $X$ that is a perfect phylogeny for $\chi$.
\end{observation}

\begin{figure}[t]
\center
\scalebox{1.3}{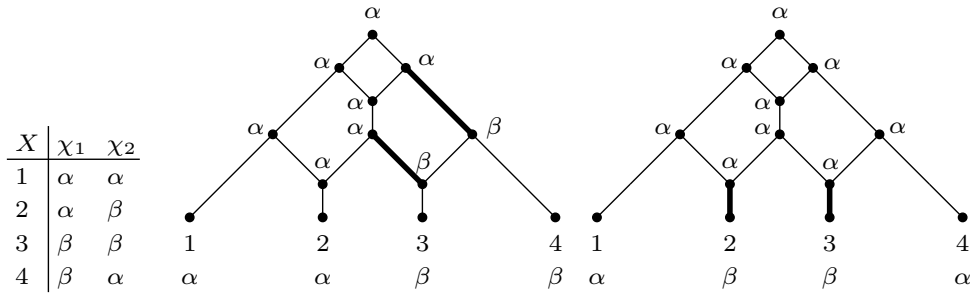}
\caption{Left: Two characters $\chi_1$ and $\chi_2$ on $X$, each with two character states $\alpha$ and $\beta$. Middle and right:  An extension $\bar\chi_1$ (resp. $\bar\chi_2$) of $\chi_1$ (resp. $\chi_2$) to the vertex set of the underlying rooted phylogenetic network $\cN$ on four leaves. Indicated by the thicker edges, note that both extensions yield two edges in $\cN$ whose two endpoints are assigned to two different character states.}
\label{fig:hardwired}
\end{figure}




\section{Parsimony on Networks}\label{sec:parsimony}
In this section, we review the different notions of parsimony on networks. In particular, we describe the {\it hardwired} and {\it softwired} notion that have been introduced by \citet{kannan12} and \citet{nakhleh05}, respectively. For the softwired notion, we describe three equivalent definitions, one of which is new to this paper. Moreover, we analyze the big parsimony problem on networks, and present new and curious properties of MP networks under both notions that  challenge the biological relevance of such networks.

\subsection{Hardwired Parsimony}\label{sec:hard}
The hardwired parsimony score of a character $\chi$ on an acyclic digraph $G$ extends the definition of the parsimony score of $\chi$ on a rooted phylogenetic tree in, possibly, the most natural way. Intuitively, the hardwired parsimony score of $\chi$ on $G$ equates to the smallest number of character-state transitions over all edges of $G$ that is required to explain $\chi$ on $G$. 

Formally, let $S=(\chi_1,\chi_2,\ldots,\chi_k)$ be a sequence of characters on $X$, and let  $G$ be an acyclic digraph with a distinguished set $X$. Then, the {\it hardwired parsimony score} of $S$ on $G$ is defined as $$PS_{hard}(S,G)=\sum_{i=1}^k\min_{\bar\chi_i}(\ch(\bar\chi_i,G)),$$ where, for each character $\chi_i$, the minimum is taken over all extensions of $\chi_i$ to $V(G)$. We note that the hardwired parsimony score of $S$ on $G$ coincides with that of the (ordinary) parsimony score \citep{fitch71} if $G$ is a (rooted) phylogenetic tree $\cT$ on $X$ and denote the latter score by $PS(S,\cT)$. Moreover, since a rooted phylogenetic network $\cN$ is a special type of acyclic digraph, the definition of the hardwired parsimony score of $S$ on $G$ naturally carries over to the hardwired parsimony score of $S$ on $\cN$. 

In practice, we are usually not given a rooted phylogenetic network. We are simply given a sequence $S$ of characters on $X$ and the aim is to find a rooted phylogenetic network $\cN$ on $X$ that has the smallest hardwired parsimony score for $S$ among all such networks, i.e. $PS_{hard}(S,\cN)\leq PS_{hard}(S,\cN')$ for each rooted phylogenetic network $\cN'$ on $X$. We refer to $\cN$ as a {\it hardwired MP network} and denote the corresponding parsimony score by $PS_{hard}(S)$. For example, Figure~\ref{fig:hardwired} shows a rooted phylogenetic network $\cN$ whose hardwired parsimony score is $PS_{hard}((\chi_1,\chi_2),\cN)=4$, where $\chi_1$ and $\chi_2$ are the two characters shown on the left-hand side of the same figure.

The first main result, Theorem~\ref{t:tree}, describes the first of our curious properties for MP networks. Let $S$ be a sequence of characters on $X$, and let $G$ and $G'$ be two acyclic digraphs with distinguished set $X$. If $G'$ can be obtained from $G$ by deleting an edge, deleting a vertex, or contracting a vertex with in-degree one and out-degree one, then it is easily checked that the hardwired parsimony score for $S$ on $G'$ is at most the hardwired parsimony score for $S$ on $G$. We summarize this result in the following observation, for which an example is shown in Figure~\ref{fig:hardwired_seq}.
\begin{observation}\label{ob-hard}
For an acyclic digraph $G$ with distinguished set $X$, deleting an edge or vertex in $G$ that is not in $X$ without disconnecting $G$, or contracting a vertex of $G$ with in-degree one and out-degree one never increases the hardwired parsimony score.
\end{observation}

\begin{figure}[t]
\center
\scalebox{1.1}{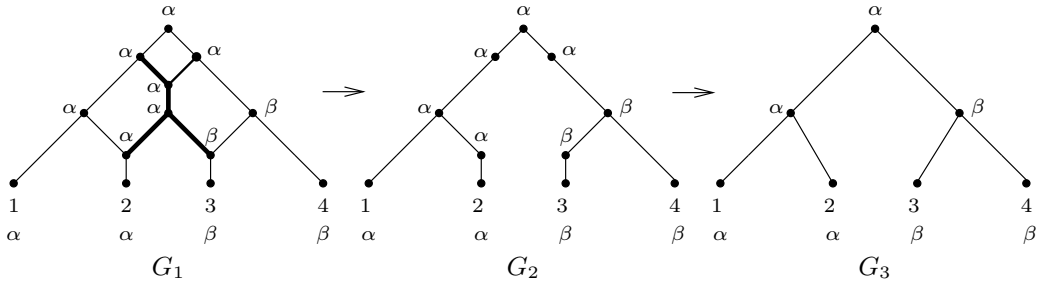}
\caption{An example to illustrate Observation~\ref{ob-hard}, where $G_2$ is obtained from $G_1$ by a sequence of edge (indicated by the thicker edges in $G_1$) and vertex deletions, and $G_3$ is obtained from $G_2$ by contracting all vertices with in-degree one and out-degree one. Note that $PS_{hard}(\chi_1,G_1)=2$ and $PS_{hard}(\chi_1,G_2)=PS_{hard}(\chi_1,G_3)=1$, where $\chi_1$ is the character shown on the left-hand side of Figure~\ref{fig:hardwired}.}
\label{fig:hardwired_seq}
\end{figure}

The next theorem follows by taking Observation~\ref{ob-hard} to an extreme for a rooted phylogenetic network $\cN$ on $X$, i.e.  deleting edges and vertices, and contracting vertices with in-degree one and out-degree one in $\cN$ until the resulting graph is a rooted phylogenetic tree on $X$.
\begin{theorem}\label{t:tree}
Let $S$ be a sequence of characters on $X$. There is always a rooted phylogenetic tree on $X$ that is a hardwired {\rm MP} network for $S$.
\end{theorem}

\noindent Theorem~\ref{t:tree} immediately follows from the next lemma which is also used in the proof of Theorem~\ref{l:displayedTrees}.

\begin{lemma}\label{l:deleteEdge}
Let $\cN$ be a rooted phylogenetic network on $X$, and let $\cT$ be a rooted phylogenetic tree on $X$ that is displayed by $\cN$. Furthermore, let $\chi$ be a character on $X$, and let $\bar{\chi}$ be an extension of $\chi$ to $V(\cN)$. Then, there exists an extension $\bar{\chi}_1$ of $\chi$ to $V(\cT)$  such that $\ch(\bar{\chi},\cN)\geq \ch(\bar{\chi}_1,\cT)$.
\end{lemma}

\begin{proof}
By the definition of displaying, $\cT$ can be obtained from $\cN$ by first deleting edges and vertices to get a tree $T$, and then contracting any resulting vertices with in-degree one and out-degree one. By construction,
$$V(\cT)\subseteq V(T)\subseteq V(\cN).$$
Let $f$ be the identity function from $V(\cT)$ to $V(T)$, and let $g$ be the identity function from $V(T)$ to $V(\cN)$. Now, let $\bar{\chi}_1$ be the extension of $\chi$ to $V(\cT)$ such that $\bar{\chi}_1(v)=\bar\chi(g(f(v)))$ for each vertex $v$ in $V(\cT)$. Let $e=(u, v)$ be an edge of $\cT$. Note that $e$ corresponds to a path in $T$ from $f(u)$ to $f(v)$. If $\bar{\chi}_1(u)\neq \bar{\chi}_1(v)$, then $e$ contributes one to ${\rm ch}(\bar{\chi}_1,\cT)$. Moreover, the edges on the path $f(u)=w_1,w_2,\ldots,w_n=f(v)$ in $T$, and therefore the edges on the path $g(f(u))=g(w_1),g(w_2),\ldots,g(w_n)=g(f(v))$ in $\cN$, collectively contribute at least one to ${\rm ch}(\bar\chi,\cN)$. Summing over all edges in $\cT$, we deduce that ${\rm ch}(\bar\chi,\cN)\ge {\rm ch}(\bar{\chi}_1,\cT)$.
\end{proof}

\noindent Following on from Theorem~\ref{t:tree}, it can be shown that each hardwired MP network $\cN$ for $S=(\chi_1,\chi_2,\ldots,\chi_k)$ that is not a phylogenetic tree has the following property. Let $v$ be a reticulation in $\cN$ and, for each $i\in\{1,2,\ldots,k\}$, let $\bar\chi_i$ be an extension of $\chi_i$ such that that $\bar\chi_1,\bar\chi_2,\ldots,\bar\chi_k$ collectively realize $PS_{hard}(S)$. Then $v$ and all its parents are assigned to the same character state. To justify this comment, assume that  there exists a character in $S$ for which $v$ and a parent, say $p$, of $v$ are assigned to two different character states. Then deleting the edge $(p,v)$ in $\cN$ decreases the hardwired parsimony score. Now, by subsequently deleting edges and vertices, and contracting vertices, we can always obtain a rooted phylogenetic tree on $X$ from $\cN$ which, by Observation~\ref{ob-hard}, has the property that its hardwired parsimony score is strictly less than that of $\cN$. This contradicts the assumption that $\cN$ is a hardwired MP network for $S$. Hence, from a biological point of view, it seems to be sensible to argue that $v$ and all of its parental species have the same genetic makeup;
thereby indicating that the associated reticulation event is possibly redundant.

Referring back to Theorem~\ref{t:tree}, it is not too difficult to see that each rooted phylogenetic tree that is a hardwired MP network $\cN$ for $S$ is, in fact, an MP tree for $S$ since, otherwise, $\cN$ is not optimal. Moreover, by slightly strengthening this fact, the next theorem uncovers an interesting property of all rooted phylogenetic trees that are displayed by a hardwired MP network for $S$.
\begin{theorem}\label{l:displayedTrees}
Let $S$ be a sequence of characters on $X$, and let $\cN$ be a hardwired {\rm MP} network for $S$. Each rooted phylogenetic tree on $X$ that is displayed by $\cN$ is an {\rm MP} tree for $S$.
\end{theorem}

\begin{proof}
Let $\cT$ be a rooted phylogenetic tree on $X$ that is displayed by $\cN$. By the optimality of $\cN$, we have $PS_{hard}(S,\cN)\leq PS_{hard}(S,\cT)$. Furthermore, applying Lemma~\ref{l:deleteEdge} to each character in $S$, we also have $PS_{hard}(S,\cN)\geq PS_{hard}(S,\cT)$. Hence,
$$PS_{hard}(S)=PS_{hard}(S,\cN)=PS_{hard}(S,\cT),$$
and so $\cT$ is a hardwired MP network for $S$. Moreover, since the (ordinary) parsimony definition for rooted phylogenetic trees coincides with the hardwired definition when restricted to rooted phylogenetic trees, it follows that $\cT$ is an MP tree for $S$. This completes the proof of the theorem.
\end{proof}


We end this section with a result on the computational complexity of the big parsimony problem on phylogenetic networks under the hardwired definition. Similar to the big parsimony problem on phylogenetic trees \citep{foulds82}, the next corollary states that it takes exponential time to compute a hardwired MP network for a sequence of characters.
\begin{corollary}\label{t:hard-hard}
Let $S$ be a sequence of characters on $X$. Finding a hardwired {\rm MP} network for $S$ is NP-hard.
\end{corollary}

\begin{proof}
To prove that the result holds, assume the contrary. Then it takes time polynomial in the size of $X$ and $S$ to calculate a hardwired MP network $\cN$ for $S$. Let $\cT$ be any rooted phylogenetic tree on $X$ that is displayed by $\cN$. By Theorem~\ref{l:displayedTrees}, $\cT$ is an MP tree for $S$. Since such a tree can be constructed from $\cN$ in polynomial time, this contradicts the fact that calculating a maximum parsimony tree for $S$ is NP-hard \citep{foulds82}. Hence, calculating a hardwired MP network for $S$ is NP-hard.
\end{proof}

\noindent On the positive side, it is however worth noting that in order to find a hardwired MP network for a sequence of characters on $X$, by Theorem~\ref{t:tree}, it suffices to search through all rooted phylogenetic trees on $X$ instead of the greatly enlarged space of all rooted phylogenetic networks on $X$~\citep{mcdiarmid14}.

\subsection{Softwired Parsimony}\label{sec:soft}
While the evolution of a set of species whose past is likely to include reticulation events can often be best represented by a phylogenetic network, the evolution of a particular gene or DNA segment can generally be described without reticulation events and therefore be represented by a phylogenetic tree. Hence, it seems plausible to assume that the evolution of a character, which is often associated with a gene or a single nucleotide, can also be represented by a tree. Using this idea, the softwired parsimony score of a character $\chi$ on a rooted phylogenetic network $\cN$ 
 is defined to be the smallest number of character-state transitions that is necessary to explain $\chi$ on any tree that is displayed by $\cN$.

Formally, we have the following definition. Let $S=(\chi_1,\chi_2,\ldots,\chi_k)$ be a sequence of characters on $X$, and let $\cN$ be a rooted phylogenetic network on $X$. Then, the {\it softwired parsimony score} of $S$ on $\cN$ is defined as $$PS_{soft}(S,\cN)=\sum_{i=1}^k\min_{\cT\in\cD(\cN)}\min_{\bar\chi_i}(\ch(\bar\chi_i,\cT))=\sum_{i=1}^k\min_{\cT\in\cD(\cN)}PS(\chi_i,\cT),$$ where, for each character $\chi_i$, the first minimum is taken over all rooted phylogenetic trees $\cT$ on $X$ displayed by $\cN$ and the second minimum is taken over all extensions of $\chi_i$ to $V(\cT)$. Similar to the previous section, it is worth noting that, if $\cN$ is a rooted phylogenetic tree, then the (ordinary) parsimony score \citep{fitch71} of $S$ on $\cN$ is equal to the softwired parsimony score of $S$ on $\cN$. Lastly, we refer to $\cN$ as a {\it softwired MP network} and denote the corresponding parsimony score by $PS_{soft}(S)$ if $\cN$ has the smallest softwired parsimony score for $S$ among all such networks, i.e. $PS_{soft}(S,\cN)\leq PS_{soft}(S,\cN')$ for each rooted phylogenetic network $\cN'$ on $X$. To illustrate, Figure~\ref{fig:softwired} shows a rooted phylogenetic network $\cN'$ with  $PS_{soft}(\chi_2,\cN')=1$, where $\chi_2$ is the character that is shown on the left-hand side of Figure~\ref{fig:hardwired}. Indeed, it is easily checked that $\cN'$ is a softwired MP network for $\chi_2$.

\begin{figure}[t]
\center
\scalebox{1.3}{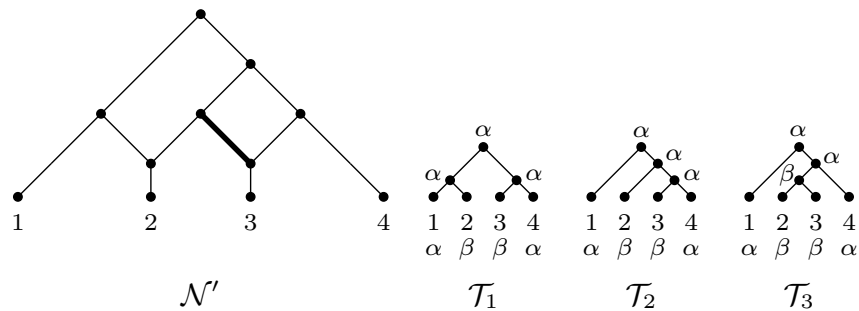}

\caption{A rooted phylogenetic network $\cN'$ on $X=\{1,2,3,4\}$ that displays the three rooted phylogenetic trees $\cT_1$, $\cT_2$, and $\cT_3$ on $X$ that are shown on the right-hand side. With $\chi_2$ being the character shown on the left-hand side of Figure~\ref{fig:hardwired}, we see that $\chi_2$ can be explained on $\cT_3$ with just one character-state transition while two character-state transitions are necessary to explain $\chi_2$ on each of $\cT_1$ and $\cT_2$. Hence, we have  $PS_{soft}(\chi_2,\cN')=1$. Moreover, to illustrate Observation~\ref{ob-soft}, note that the rooted phylogenetic network $\cN'$ can be obtained from the network $\cN$ that is shown in Figure~\ref{fig:prelim} by adding an edge that joins two new non-leaf vertices (indicated by the thicker edge). Since $\cN$ does not display $\cT_3$, we have $PS_{soft}(\chi_2,\cN)>PS_{soft}(\chi_2,\cN')$.
}
\label{fig:softwired}
\end{figure}

We next describe our second curious property for MP networks. Let $\cN$ and $\cN'$ be two rooted phylogenetic networks on $X$ such that $\cN'$ can be obtained from $\cN$ by subdividing two edges and adding a new edge joining the two new vertices. Since the collection of rooted phylogenetic trees displayed by $\cN$ is a subset of the collection of rooted phylogenetic trees that are displayed by $\cN'$, the next observation, which is in stark contrast to Observation~\ref{ob-hard}, is an immediate consequence of the definition of the softwired parsimony score.

\begin{observation}\label{ob-soft}
Adding an edge joining two new non-leaf vertices to a rooted phylogenetic network never increases the softwired parsimony score. 
\end{observation}
\noindent This observation was first mentioned by \citet{nakhleh05}, who noticed that networks with a large number of reticulations tend to have a smaller parsimony score. An example to illustrate Observation~\ref{ob-soft} is shown in Figure~\ref{fig:softwired}. 

Perhaps surprisingly, in comparison to what happens under the hardwired notion, the next theorem states that solving the big parsimony problem on networks under the softwired definition for a sequence $S$ of characters on $X$ is not NP-hard. Intuitively, this can be justified by noting that it takes polynomial time to construct a rooted phylogenetic network $\cN$ that displays a perfect phylogeny  on $X$ for each character in $S$. It then follows that $\cN$ is a softwired MP network for $S$.

In the proof of the next theorem, we make use of a construction in \citet{francis15}. In particular, the authors describe a construction of a rooted phylogenetic network $\cN$ on $X$ that displays all rooted binary phylogenetic trees on $X$.  Additionally, they have shown that $\cN$ can be constructed from a rooted binary phylogenetic tree on $X$ by adding $\frac 1 2n(n-1)^2$ edges to $\cT$, where each such edge joins two vertices that subdivide edges in $\cT$ and $n=|X|$. Note that, as $\cT$ has $O(n)$ edges, $\cN$ has $O(n^3)$ edges. In what follows, we call a rooted binary phylogenetic network on $X$ a {\it universal network} on $X$ if it displays all rooted binary phylogenetic trees on $X$.

\begin{theorem}\label{soft-hard}
Let $S$ be a sequence of characters on $X$. Finding a softwired {\rm MP} network for $S$ is solvable in time polynomial in the size of $X$.
\end{theorem}

\begin{proof}
Let $n=|X|$, and suppose that $S=(\chi_1,\chi_2,\ldots,\chi_k)$ is a sequence of characters on $X$. Let $\cN$ be a universal network on $X$ whose number of edges is polynomial in $n$. By the paragraph prior to this theorem, such a network exists \citep{francis15}. Hence, $\cN$ can be constructed in time polynomial in $n$.

We complete the proof by showing that $\cN$ is a softwired MP network for $S$. For each $i\in\{1,2,\ldots,k\}$, let $r_i$ denote the number of distinct character states of $\chi_i$, and let $T_i$ be the unique rooted tree with exactly $r_i+1$ internal vertices that has the following properties. The set of vertices of $T_i$ with out-degree zero is precisely $X$, the root of $T_i$ is adjacent to $r_i$ internal vertices, and two elements in $X$, say $\ell$ and $\ell'$, are adjacent to the same internal vertex of $T_i$ if and only if $\chi_i(\ell)=\chi_i(\ell')$. Now, obtain a rooted phylogenetic tree $\cT_i$ on $X$ from $T_i$ by contracting each vertex with in-degree one and out-degree one. Let $\cT_i'$ be any binary refinement of $\cT_i$. It is easily checked that $PS(\chi_i,\cT_i)=PS(\chi_i,\cT_i')=r_i-1$. In particular, $\chi_i$ is homoplasy-free on $\cT_i'$ and, hence, by Proposition~5.1.3 of~\citet{semple03}, $\cT_i'$ is an MP tree for $\chi_i$. Furthermore, by construction, $\cN$ displays $\cT_i'$. It now follows that $$PS_{soft}(S,\cN)=\sum_{i=1}^k r_i-1=PS_{soft}(S).$$ This completes the proof of the theorem.
\end{proof}

From a practical viewpoint, the construction in the proof of Theorem~\ref{soft-hard} implies that one can construct a softwired MP network for an arbitrary sequence $S$ of characters on $X$ without looking at the data by simply constructing a universal network on $X$. 

We end this subsection with two equivalent ways of viewing the softwired notion of the parsimony score of a sequence of characters on a phylogenetic network. The first is due to~\citet{fischer}. Let $S$ be a sequence of characters on $X$. Furthermore, let $\cN$ be a rooted phylogenetic network on $X$, and let $T$ be a rooted tree with a distinguished set $X$. Note that $T$ is not necessarily a phylogenetic tree.~Then $T$ is called a {\it switching} of $\cN$ if it can be obtained from $\cN$ by deleting, for each reticulation $v$, all but one edge directed into $v$. It is easily checked that each rooted phylogenetic tree on $X$ that is displayed by $\cN$ can be obtained from a switching of $\cN$ by repeated applications of the following two operations: deleting unlabeled vertices of degree one, and contracting vertices with in-degree one and out-degree one. Conversely, each switching of $\cN$ can be transformed into a rooted phylogenetic tree on $X$ that is displayed by $\cN$ by repeated applications of the same two operations.

Now, let $\cS(\cN)$ denote the set of all switchings of a rooted phylogenetic network $\cN$. The next theorem allows us to work with $\cS(\cN)$ instead of the set of all trees that are displayed by $\cN$ to compute  $PS_{soft}(S,\cN)$.

\begin{theorem}[Lemma 4.5 of \citet{fischer}]\label{l:switch}
Let $S=(\chi_1,\chi_2,\ldots,\chi_k)$ be a sequence of characters on $X$, and let $\cN$ be a rooted phylogenetic network on $X$. Then, $$PS_{soft}(S,\cN)=\sum_{i=1}^k\min_{T\in\cS(\cN)}\min_{\bar\chi_i}(\ch(\bar\chi_i,T)),$$ where, for each character $\chi_i$, the first minimum is taken over all switchings $T$ of $\cN$ and the second minimum is taken over all extensions of $\chi_i$ to $V(T)$.
\end{theorem}

The second equivalence is new to this paper and requires a new definition. Let $\chi$ be a character on $X$, and let $\cN$ be a rooted phylogenetic network on $X$. Furthermore, let $\bar\chi$ be an extension of $\chi$ to $V(\cN)$. For a reticulation edge $(u,v)$ of $\cN$, we say that $(u,v)$ is a {\it negligible edge under $\bar\chi$} if $\bar\chi(u)\ne\bar\chi(v)$ but there exists a parent $p$ of $v$ in $\cN$ such that $\bar\chi(p)=\bar\chi(v)$. We use ${\rm n}(\bar\chi,\cN)$ to denote the number of negligible edges under $\bar\chi$ in $\cN$. For example, in Figure~\ref{fig:hardwired}, the extension $\bar\chi_1$ of $\chi_1$ to the vertex set of the rooted phylogenetic network $\cN$ shown in the middle of the same Figure has ${\rm n}(\bar\chi_1,\cN)=1$. The next theorem 
shows how a hardwired-type approach that considers the number of negligible edges can be used to compute the softwired parsimony score.

\begin{theorem}\label{t:newDef}
Let $S=(\chi_1,\chi_2,\ldots,\chi_k)$ be a sequence of characters on $X$, and let $\cN$ be a rooted phylogenetic network on $X$. Then $$PS_{soft}(S,\cN)=\sum_{i=1}^k\min_{\bar\chi_i}(\ch(\bar\chi_i,\cN)-{\rm n}(\bar\chi_i,\cN)),$$ where, for each character $\chi_i$, the minimum is taken over all extensions of $\chi_i$ to $V(\cN)$.
\end{theorem}

\begin{proof}
To establish the theorem, it suffices to show that the result holds when $S$  consists of a single character $\chi$, that is $$PS_{soft}(\chi,\cN)=\min_{\bar\chi}(\ch(\bar\chi,\cN)-{\rm n}(\bar\chi,\cN)).$$ Throughout the proof, we make use of Theorem~\ref{l:switch} and consider the set of all switchings of $\cN$ to compute $PS_{soft}(\chi,\cN)$.
Let $\bar{\chi}_1$ be an extension of $\chi$ to $V(\cN)$ such that $\ch(\bar{\chi}_1,\cN)-{\rm n}(\bar{\chi}_1,\cN)=\min\limits_{\bar\chi}(\ch(\bar\chi,\cN)-{\rm n}(\bar\chi,\cN))$. Furthermore, let $T_1$ be a switching of $\cN$ such that for each reticulation $v$ of $\cN$ that has a parent $p_v$ with $\bar{\chi}_1(v)=\bar{\chi}_1(p_v)$, the edge $(p_v,v)$ is an edge of $T_1$. Since $V(T_1)=V(\cN)$, it follows that, by taking the identity function from $V(\cN)\rightarrow V(T_1)$, we can view $\bar{\chi}_1$ as an extension of $\chi$ to $V(T_1)$. Hence,
\begin{eqnarray}\label{eq:one}
\min_{\bar\chi}(\ch(\bar\chi,\cN)-{\rm n}(\bar\chi,\cN))& =&\ch(\bar{\chi}_1,\cN)-{\rm n}(\bar{\chi}_1,\cN)\nonumber\\
&\geq &\ch(\bar{\chi}_1,T_1)\geq PS_{soft}(\chi,\cN),
\end{eqnarray}
where the second inequality follows from Theorem~\ref{l:switch}.

Now, by Theorem~\ref{l:switch}, there is a switching $T_2$ of $\cN$ and an extension $\bar{\chi}_2$ of $\chi$ to $V(T_2)$ such that $\ch(\bar{\chi}_2,T_2)=PS_{soft}(\chi,\cN)$. We next show that $\bar{\chi}_2$ can always be chosen so that, for each edge $(u,v)$ in $T_2$ that is a reticulation edge in $\cN$, we have $\bar{\chi}_2(u)=\bar{\chi}_2(v)$. 

Let $(u,v)$ be an edge in $T_2$ that is a reticulation edge in $\cN$ and whose two endpoints are assigned to two different character states, i.e. $\bar{\chi}_2(u)\ne\bar{\chi}_2(v)$. Furthermore, let $\bar{\chi}_3$ be the extension of $\chi$ to $V(T_2)$ such that $\bar{\chi}_3(v)=\bar{\chi}_2(u)$, and $\bar{\chi}_3(w)=\bar{\chi}_2(w)$ for each vertex $w$ of $T_2$ other than $v$. Recalling that, by the definition of a rooted phylogenetic network, $v$ has exactly one child, it now follows that the contribution of the edges incident with $v$ in $T_2$ to $\ch(\bar{\chi}_2,T_2)$ is at least the contribution of those edges to $\ch(\bar{\chi}_3,T_2)$. 
In particular, by the optimality of $\bar{\chi}_2$, we have $\ch(\bar{\chi}_2,T_2)=\ch(\bar{\chi}_3,T_2)$. Setting $\bar{\chi}_2$ to be $\bar{\chi}_3$ and repeating this argument for each edge in $T_2$ that is a reticulation edge in $\cN$ and whose two endpoints are assigned to two different character states, it follows that we eventually obtain an extension $\bar{\chi}_2$ of $\chi$ to $V(T_2)$ that realizes $PS_{soft}(\chi,\cN)$ and has the property that, for each edge $(u,v)$ of $T_2$ that is a reticulation edge in $\cN$, we have $\bar{\chi}_2(u)=\bar{\chi}_2(v)$. Furthermore, as $T_2$ is a spanning tree of $\cN$, the extension $\bar{\chi}_2$ is also an extension of $\chi$ to $V(\cN)$. It is now easily checked that each reticulation edge in $\cN$ that is not an edge in $T_2$ either has two endpoints that are assigned to the same character state or is a negligible edge under $\bar{\chi}_2$, and so
\begin{eqnarray}\label{eq:two}
PS_{soft}(\chi,\cN)=\ch(\bar{\chi}_2,T_2)&=&\ch(\bar{\chi}_2,\cN)-{\rm n}(\bar{\chi}_2,\cN)\nonumber\\
&\geq & \min_{\bar\chi}(\ch(\bar\chi,\cN)-{\rm n}(\bar\chi,\cN)).
\end{eqnarray}
Combining the two inequalities (\ref{eq:one}) and (\ref{eq:two}) establishes the theorem.
\end{proof}

\noindent Intuitively, in the second equivalence, we do not `penalize' reticulation edges $(u, v)$ directed into a reticulation $v$ whose endpoints are assigned to different character states provided there is at least one reticulation edge $(p, v)$ whose endpoints are assigned to the same state.

\section {Connections between Maximum Parsimony and Maximum Likelihood on phylogenetic networks} \label{sec:MPML}

When considering MP on networks, a natural question is which properties of MP on trees still hold in the more general setting of phylogenetic networks. One well-known property of MP on trees is its equivalence with maximum likelihood (ML) \citep{TS} under the symmetric $r$-state model with `no common mechanism'. We will briefly introduce this model and the equivalence result here before we analyze its parallels to phylogenetic networks.

Before we state the results, we introduce some definitions. Recall that the symmetric $r$-state model, which is also often called {\it $N_{r}$-model} (and Jukes Cantor model for $r=4$ character states~\citep{JC}), is defined as follows. Let $\cN$ be a rooted phylogenetic network, and let $\{c_{1},c_2, \ldots, c_{r}\}$ be $r$ distinct character states with $r\geq 2$. The $N_{r}$-model assumes a uniform distribution of states at the root of $\cN$ and equal rates of substitutions between any two distinct character states \citep{N}. Under the $N_{r}$-model, we denote by $p(e)$ the probability that a substitution of a character state $c_{i}$ by another character state $c_{j}$ occurs on some edge $e \in E(\cN)$  for $c_{i} \neq c_{j}$. Furthermore, let $q(e)=1-(r-1)p(e)$ denote the probability that no substitution occurs on edge $e$. Then, in the $N_{r}$-model, we have $ 0 \leq p(e) \leq \frac{1}{r}$ for all $e \in E(\cN)$ and $(r-1)p(e)+q(e)=1$. Note that the $N_r$-model is time-reversible, i.e. it does not matter where the root of a network is placed, and the rate of change from state $c_i$ to $c_j$ is the same as that from $c_j$ to $c_i$. Lastly, we assume that, if a sequence consists of at least two characters, then the different characters have evolved under \textit{no common mechanism} \citep{TS}. This means that the substitution probabilities on the edges of the underlying network $\cN$ may be different for each character in the sequence without any correlation between them.


We will now turn our attention to likelihood concepts. Let $\cT$ be a rooted phylogenetic tree and let $\chi$ be a character on $X$. Recall that the probability $P(\chi | \cT,P^\cT)$ of $\chi$, for a given probability vector $P^\cT$ for character-state transitions on the edges of $\cT$, is the probability that a root state evolves along $\cT$ to the joint assignment
of leaf states induced by $\chi$. Furthermore, we have $P(\chi | \cT,P^\cT)=\sum_{\bar{\chi}}P(\bar{\chi} \mid \cT, P^\cT)$, i.e. the likelihood of $\chi$ on $\cT$ can be calculated as the sum of the likelihoods of all possible extensions of $\chi$ to $V(\cT)$ \citep{felsenstein1981}. The {\it ML} of $\chi$ on $\cT$, denoted by $\max \,P(\chi|\cT)$,  is the value of $P(\chi|\cT,P^\cT)$  maximized over all possible assignments of substitution probabilities $P^\cT$, i.e.  $\max \,P(\chi|\cT) = \max\limits_{P^\cT}P(\chi|\cT,P^\cT)$. Moreover, the trees for which $\max P(\chi|\cT)$ is maximum are called {\em ML trees}.\\ \par

We are now in a position to state the equivalence result of MP and ML for trees.

\begin{theorem}[Theorem 5 of \citet{TS}]
\label{T1}
Let $\cT$ be a phylogenetic tree on $X$, and let $S=(\chi_{1},\chi_2,\ldots, \chi_{k})$ be a sequence of $r$-state characters on $X$. Then, under the $N_{r}$-model with no common mechanism,
\begin{equation}\label{eq:TS}
\max P(S \mid \cT)= r^{-PS(S,\cT)-k}.
\end{equation}
Thus, {\rm ML} and {\rm MP} both choose the same tree(s).
\end{theorem}

Note that Theorem \ref{T1} not only implies that both methods choose the same optimal sets of rooted phylogenetic trees, but rather that both methods induce the same ranking of trees. This means that whenever a phylogenetic tree $\cT$ on $X$ has a lower parsimony score than another such tree $\cT'$, then Equation (\ref{eq:TS}) implies that the likelihood of $\cT$ is higher than that of $\cT'$, which means that $\cT$ will both be more parsimonious and more likely than $\cT'$. We will later  use this fact to directly establish a similar equivalence result for the softwired parsimony setting on phylogenetic networks.

\subsection{Softwired Likelihood Functions}
We now turn to likelihood on phylogenetic networks. Let $\cN$ be a rooted phylogenetic network on $X$, and let $\chi$ be an $r$-state character on $X$. Let $P^\cN=(p(e_{i}): e_{i} \in E(\cN))$ denote the vector of the probabilities $p(e_i)$ of a character-state transition on the edges of $\cN$ under the $N_{r}$-model. Furthermore, let $\cT$ be a rooted phylogenetic tree on $X$ that is displayed by $\cN$, and let $\cE^\cT$ be an embedding of $\cT$ in $\cN$. Note that $\cE^\cT$ is not necessarily unique. We define a substitution probabilities vector $P^{\cE^\cT}=(p'(e'_{j}): e'_{j} \in E(\cT))$ assigned to the edges of $\cT$ as follows. For each edge $e'_j$ in $\cT$ that corresponds to a unique edge $e_j$ in $\cN$ (and $\cE^\cT$), we set $p'(e'_{j})=p(e_{j})$. Otherwise, $e'_j$ corresponds to a path of edges in $\cN$ (and $\cE^\cT$). If $e'_{j}$ corresponds to exactly two edges, say $e_i$ and $e_k$, in $\cN$,
we set 
\begin{equation*}
p'(e'_{j})= p(e_{i})+p(e_{k})-r\cdot p(e_i)p(e_k).
\end{equation*}
This definition considers the amount of change on both edges $e_i$ and $e_k$, which correspond to $e'_j$ as well as the $r$ possible situations where a change on $e_k$ undoes a change on $e_i$ so that there is no change occurring on $e_j'$.
This last part is subtracted.  Moreover, $p'(e'_{j})$ is equal to $0$ precisely when both values $p(e_i)$ and $p(e_k)$ are $0$, else it is positive. If $e'_{j}$ corresponds to a path of $l$ edges in $\cN$ with $l>2$, we iteratively apply the above equation $l-1$ times. We call $P^{\cE^\cT}$ a {\em restriction} of $P^\cN$ to  $\cT$ under the $N_{r}$-model. Furthermore, we denote  by $P^{\cT}$ a restriction of $P^\cN$ to  $\cT$ for which the probability of observing $\chi$ given $\cT$ and $P^\cT$ is maximized over all embeddings of $\cT$ in $\cN$, i.e. $$P(\chi\mid \cT, P^\cT)=\max_{P^{\cE^\cT}}P(\chi\mid \cT, P^{\cE^\cT}).$$

\begin{figure}[t]
\center
\scalebox{0.6}{ \includegraphics{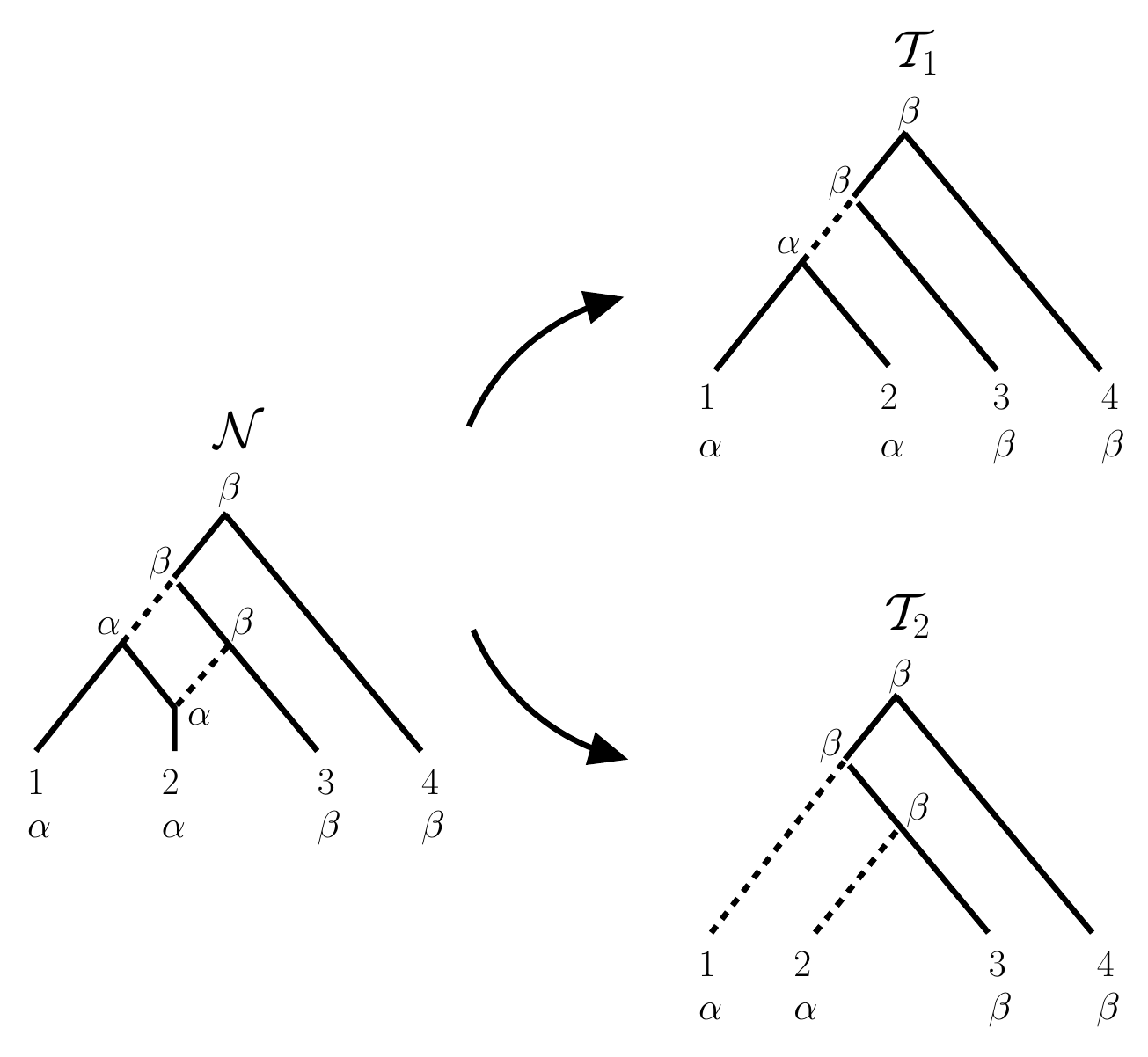} }
\caption{A rooted phylogenetic network $\cN$ on $X=\{1,2,3,4\}$ with one reticulation and the two phylogenetic trees $\cT_1$ and $\cT_2$ it displays. The character $\chi$ that is associated with the leaf labels in $\cN$ is also depicted together with a most parsimonious extension to the inner vertices. The dashed edges represent character-state transitions.}\label{fig:1}
\end{figure}

We next show that, for a frequently-used likelihood function, the equivalence of parsimony and likelihood on phylogenetic networks no longer holds. Let $\cN$ be a rooted phylogenetic network on $X$, let $\chi$ be an $r$-state character on $X$, and suppose that  $P^\cN$ 
is a vector of substitution probabilities on the edges of $\cN$ under the $N_{r}$-model with no common mechanism.  Furthermore, let $\cT$ be a rooted phylogenetic tree on $X$ that is displayed by $\cN$. We denote by $P(\cT \mid \cN,\chi)$  the probability that $\cT$ is chosen amongst all trees that are displayed by $\cN$. The above-mentioned likelihood function is the following, which can be found, for example, in \citet{NA}.
\begin{equation*}
 P_{w\textnormal{-}soft}(\chi \mid \cN,P^\cN)= \max\limits_{\cT \in \cD(\cN)}(P(\cT \mid \cN,\chi)\cdot P(\chi \mid \cT, P^\cT)).
\end{equation*}
We call this the {\it weighted softwired likelihood}, and the maximum of the weighted softwired likelihood over all probability assignments $P^\cN$ will be denoted by $\max  P_{w\textnormal{-}soft}(\chi \mid \cN)$. Biologically, it makes sense to distinguish between trees which are likely to be chosen and those which are not. However, softwired parsimony and weighted softwired likelihood are not equivalent on phylogenetic networks. More formally, we show, by means of a counterexample that consists of a single character, that $$\max P_{w\textnormal{-}soft}(\chi\mid\cN)= r^{-PS_{soft}(\chi,\cN)-k}$$ does not hold.

Consider the rooted phylogenetic network $\cN$ and the $2$-state character $\chi$ on the leaves of $\cN$ shown in Figure~\ref{fig:1}. In the same figure, the two rooted phylogenetic trees $\cT_1$ and $\cT_2$ on $X$ are precisely the trees displayed by $\cN$. Here, we have $PS(\chi,\cT_1)=1$ and $PS(\chi,\cT_2)=2$ and, hence $PS_{soft}(\chi,\cN)=1$. Moreover, assuming the $N_2$-model with no common mechanism, by Theorem \ref{T1}, we have
\begin{equation}\label{eq:T1}
\max P(\chi\mid\cT_1)=r^{-PS(\chi,\cT_1)-k}=2^{-1-1}=\frac{1}{4}
\end{equation}
and 
\begin{equation}\label{eq:T2}
\max P(\chi\mid\cT_2)=r^{-PS(\chi,\cT_2)-k}=2^{-2-1}=\frac{1}{8},
\end{equation}
where $k$ is the number of characters under consideration.
For the maximum of the weighted softwired likelihood, we have
\begin{align*}
\max P_{w\textnormal{-}soft}(\chi \mid \cN)&= \max\limits_{\cT\in \{\cT_1,\cT_2\}} (P(\cT\mid \cN,\chi)\cdot \max\limits P(\chi\mid\cT))\\
&= \max \left\{P(\cT_1\mid\cN,\chi) \cdot \frac{1}{4},P(\cT_2\mid\cN,\chi) \cdot \frac{1}{8} \right\},
\end{align*}
where the last equality follows from Equations~(\ref{eq:T1}) and~(\ref{eq:T2}).

Now, if we assume, for example, that $\cT_2$ is chosen three times as often as $\cT_1$, i.e. $P(\cT_1\mid\cN,\chi)=\frac{1}{4}$ and $P(\cT_2\mid\cN,\chi)=\frac{3}{4}$, then we have  $$\max P_{w\textnormal{-}soft}(\chi \mid \cN)= \max \left\{\frac{1}{4} \cdot \frac{1}{4},\frac{3}{4} \cdot \frac{1}{8} \right\} = \frac{3}{32}.$$
Not only is this weighted softwired ML value unequal to $r^{-PS_{soft}(\chi,\cN)-k}=\frac{1}{4}$, it is also achieved by tree $\cT_2$, whereas $\cT_1$ is strictly better than $\cT_2$ in the softwired parsimony sense. Consequently, under the weighted definition of softwired likelihood, the equivalence between parsimony and likelihood on networks fails. 

Next, we consider a second softwired likelihood concept on phylogenetic networks which was introduced in \citet{barryhartigan} and analyzed by \citet{steelpenny}. We call this concept the {\em softwired pseudo-likelihood}. An explanation of why we call it pseudo-likelihood is given later. Let $\chi$ be a character on $X$. Given a rooted phylogenetic network $\cN$ on $X$ and a vector $P^\cN$ of substitution probabilities on the edges of $\cN$, we define
the {\it softwired pseudo-likelihood} $P_{soft}(\chi \mid \cN,P^\cN)$ of $\chi$ to be
$$P_{soft}(\chi \mid \cN,P^\cN)= \max_{\cT \in \cD(\cN)} P(\chi \mid \cT, P^\cT)=\max_{\cT \in \cD(\cN)} \sum_{\bar{\chi}}P(\bar{\chi} \mid \cT, P^\cT),$$ 
where the maximum is taken over all rooted phylogenetic trees $\cT$ on $X$ that are displayed by $\cN$ and the summation is taken over all extensions $\bar{\chi}$ of $\chi$ to $V(\cT)$. 

Now, the \textit{softwired pseudo-ML} of $\chi$ on $\cN$ as the maximum value of $P(\chi \mid \cT, P^\cT)$ of the most likely rooted phylogenetic tree on $X$ which is displayed by $\cN$. That is, the softwired pseudo-ML is defined by
\begin{equation*}
 \max P_{soft}(\chi \mid \cN)= \max_{\cT \in \cD(\cN)} \max_{P^\cT}P(\chi \mid \cT, P^\cT)=\max_{\cT \in \cD(\cN)} \max_{P^\cT}\sum_{\bar{\chi}}P(\bar{\chi} \mid \cT, P^\cT),
\end{equation*}
where, for a rooted phylogenetic tree $\cT$ displayed by $\cN$, the inner maximum is taken over all vectors of substitution probabilities on the edges of $\cT$ under the $N_{r}$-model.
A (not necessarily unique) \textit{softwired pseudo-ML network} of $\chi$ is a network for which the  softwired pseudo-ML is maximum, i.e. $$\arg\max_{\cN}[\max P_{soft}(\chi \mid \cN)].$$

Note that, by definition of the $N_{r}$-model with no common mechanism, the softwired pseudo-ML for a sequence of $r$-state characters $S=(\chi_{1},\chi_2,\ldots,\chi_{k})$ on $X$ can be calculated as the product of the pseudo-likelihoods of the individual characters due to independence. Hence, we have $$\max P_{soft}(S\mid \cN)=\prod_{i=1}^{k} \max P_{soft}(\chi_{i} \mid \cN).$$ 

It is worth noting that the above definition of a pseudo-likelihood on a phylogenetic network $\cN$ on $X$ does not incorporate a probability distribution on the trees that are displayed by $\cN$.
This is the reason, why we refer to it as pseudo-likelihood. In fact, if one sums up the softwired pseudo-likelihoods $P_{soft}(\chi \mid \cN,P^\cN)$ over all possible characters $\chi$ on $X$, the sum might be larger than 1. However, this pseudo-likelihood has been discussed before in a different context (e.g. in \citet{barryhartigan, steelpenny}), and it turns out to be strongly related to the softwired parsimony notion for phylogenetic networks. Specifically, using Theorem~\ref{T1} and the fact that not only the optimal trees are the same, but the entire ranking induced by parsimony and likelihood is identical, we next show that softwired MP and  softwired pseudo-ML on networks are equivalent.


\begin{theorem}[Equivalence of softwired MP and softwired pseudo-ML for networks]\label{T2}
Let $\cN$ be a rooted phylogenetic network on $X$, and let $S=(\chi_{1},\chi_2,\ldots, \chi_{k})$ be a sequence of $r$-state characters on $X$. Then, under the $N_{r}$-model with \textit{no common mechanism},
\begin{equation*}
\max P_{soft}(S \mid \cN)=r^{-PS_{soft}(S,\cN)-k}.
\end{equation*}
Thus, softwired {\rm MP} and softwired pseudo-{\rm ML} both choose the same network(s).
\end{theorem}
\begin{proof} We first consider the case $k=1$, i.e. $S=(\chi_1)$. By Theorem~\ref{T1} and recalling the definition of the softwired parsimony score on $\cN$, we have
\begin{align} \max P_{soft}(\chi_1\mid \cN) &= \max\limits_{\cT\in \cD(\cN)} \max\limits_{P^\cT} P(\chi_1\mid \cT,P^\cT) \nonumber \\\nonumber
&=  \max\limits_{\cT\in \cD(\cN)} \max P(\chi_1\mid \cT) \\\nonumber
&=  \max\limits_{\cT\in \cD(\cN)} r^{-PS(\chi_1,\cT)-1}\\\nonumber
&=  r^{-PS_{soft}(\chi_1,\cN)-1}.
\end{align}

Now, for a sequence $S=(\chi_{1},\chi_2,\ldots, \chi_{k})$ of characters, we have
\begin{align*} 
\max P_{soft}(S\mid\cN)  &=  \prod\limits_{i=1}^{k} r^{-PS_{soft}(\chi_i,\cN)-1} \\
& =  r^{-\sum\limits_{i=1}^{k} (PS_{soft}(\chi_i,\cN)+1)}\\
&=  r^{- PS_{soft}(S,\cN)-k},
 \end{align*}
 where the first equality follows from the fact the characters are independent under the no common mechanism model and the third equality follows again from the definition of the softwired parsimony score on $\cN$.
This completes the proof.
 \end{proof}



\subsection{A Hardwired Likelihood Functions}

In this section, we analyze a hardwired notion of likelihood on networks.
Let $\cN$ be a rooted phylogenetic network on $X$, and let $\bar{\chi}$ be an extension of an $r$-state character $\chi$ on $X$ to $V(\cN)$. Furthermore, let $P^\cN=(p(e) : e \in E(\cN))$ be a substitution probabilities vector assigned to edges of $\cN$ under the $N_{r}$-model. 
We set the likelihood of $\bar{\chi}$ on $\cN$ given $P^\cN$ to be $$P(\bar{\chi} \mid \cN, P^\cN)= \frac{1}{r}\prod\limits_{\genfrac{}{}{0pt}{1}{e=(u,v):}{ \bar{\chi}(u)\neq\bar{\chi}(v)}}p(e)\prod\limits_{\genfrac{}{}{0pt}{1}{e=(u,v):}{ \bar{\chi}(u)=\bar{\chi}(v)}}q(e),$$ where the first product considers all edges in $E(\cN)$ whose two endpoints are assigned to two distinct character states and the second product considers all edges in $E(\cN)$ whose two endpoints are assigned to the same character state.
Then, the {\it hardwired pseudo-likelihood} of observing $\chi$ on $\cN$ for a given $P^\cN$ under the $N_{r}$-model is defined as
\begin{equation*}
P_{hard}(\chi \mid \cN,P^\cN)= \sum_{\bar{\chi}}P(\bar{\chi} \mid \cN, P^\cN),
\end{equation*}
where the summation is taken over all extensions $\bar{\chi}$ of $\chi$ to $V(\cN)$. 

Now, the \textit{hardwired pseudo-ML} of $\chi$ on $\cN$, denoted by $\max P_{hard}(\chi \mid \cN)$, is the maximum of $P_{hard}(\chi \mid \cN,P^\cN)$ over all $P^\cN$. Hence, $$\max P_{hard}(\chi \mid \cN)=\max_{P^\cN} P_{hard}(\chi \mid \cN,P^\cN).$$ 
Finally, a (not necessarily unique) \textit{hardwired pseudo-ML network} of $\chi$ is a network for which the hardwired pseudo-ML is maximum, i.e. $$\arg\max_{\cN}[\max P_{hard}(\chi \mid \cN)].$$ 
\\ \par

As for softwired, the hardwired maximum pseudo-likelihood score for a sequence of characters $S=(\chi_{1},\chi_2,\ldots,\chi_{k})$ on $X$ can be calculated as the product of the pseudo-likelihoods of the individual characters due to independence, i.e. $$\max P_{hard}(S\mid \cN)= \prod_{i=1}^{k} \max P_{hard}(\chi_{i} \mid \cN).$$

As for parsimony, for a rooted phylogenetic tree $\cT$ on $X$, we remark that the softwired and the hardwired definitions of ML on networks are equal and they also coincide with $\max P(S \mid \cT)$. Thus, we have $$\max P_{soft}(S \mid \cT)=\max P_{hard}(S \mid \cT)=\max P(S \mid \cT).$$ 

Moreover, we have the following equivalence result.

\begin{theorem}[Equivalence of hardwired MP and hardwired pseudo-ML for networks]\label{T3}
Let $\cN$ be a rooted phylogenetic network on $X$, and let $S=(\chi_{1},\chi_2,\ldots, \chi_{k})$ be a sequence of $r$-state characters on $X$. Then, under the $N_{r}$-model with \textit{no common mechanism},
\begin{equation*}
\max P_{hard}(S \mid \cN)=r^{-PS_{hard}(S,\cN)-k}.
\end{equation*}
Thus, hardwired {\rm MP} and hardwired pseudo-{\rm ML} both choose the same network(s).
\end{theorem}

The proof of Theorem \ref{T3} is a rather technical generalization of the proof of Theorem \ref{T1} in \citet{TS}. In particular, the proof exploits the fact that, if two leaves of a rooted phylogenetic network $\cN$ are in a different character state, then all paths in $\cN$ that connect the two leaves contain at least one edge whose two endpoints are assigned to two different states. We omit the details of this proof.

We end this section with a remark. Recall that the likelihoods of all characters on an arbitrary phylogenetic tree sum up to one. Since a  phylogenetic network $\cN$ can be obtained from some phylogenetic tree $\cT$ by adding edges, it follows that the the  likelihoods of all characters on $\cN$ may sum up to a value that is strictly less than one because each likelihood on $\cT$ will be multiplied with the substitution probability on each additional edge in $\cN$. This is the reason, why we refer to $P_{hard}(\chi \mid \cN,P^\cN)$ as  pseudo-likelihood.

\section{Conclusion}\label{sec:conclu}
The small parsimony problem on networks has recently attracted considerable attention. In particular, several related complexity questions have been settled \citep{fischer,jin09,nguyen07} and exact algorithms and heuristics \citep{fischer, kannan12, kannan14} to tackle this problem have been proposed. In contrast, the big parsimony problem has so far only been mentioned in one article \citep{wheeler15}, where formal proofs were omitted. Yet, the big problem is exactly what is ultimately of interest to evolutionary biologists who wish to reconstruct a rooted phylogenetic network from molecular data under a parsimony framework. 

In this paper, we have presented the first formal analysis of MP networks and uncovered several curious properties of such networks, including interesting parallels to functions that resemble likelihood functions. Depending on whether one reconstructs an MP network under the hardwired or softwired framework, it is potentially either overly simple (under hardwired) or overly complex (under softwired) in terms of the number of reticulations. Consequently, under both notions, the biological relevance of MP networks is challenged. In particular, the results in this paper show that neither hardwired nor softwired MP can distinguish between evolutionary histories that are best represented by a phylogenetic tree and histories that are best represented by a network.
It suggests that we need to reconsider the definition of parsimony on networks and to develop a new or improved framework. 

One such improvement, that we propose, is to consider an extension of the softwired parsimony definition that computes the parsimony score of a sequence $S$ of characters on a rooted phylogenetic network $\cN$ by first computing $PS_{soft}(S,\cN)$ and then  increasing this score by a certain user-defined `penalty' for each reticulation in $\cN$. Unless the penalty is set to zero, an MP network for $S$ under this new definition is unlikely to have a high number of reticulations and, similarly, unless the penalty is set to infinity,  such a network for $S$ is unlikely to be a tree. Since evolutionary biologists often have valuable information at their fingertips as to whether the expected amount of reticulation is significant or not for a certain data set, this information can be used to compute parsimonious networks that are biologically more meaningful than those reconstructed under the hardwired or softwired definition. In particular, if a high amount of reticulation is expected (e.g. as for certain groups of bacteria or plants) the penalty should be smaller than in the case for when one expects the evolutionary history to be almost tree-like. 

Concerning the parallels of both the softwired and hardwired parsimony concepts to likelihood concepts on phylogenetic networks, we showed in the previous section that the equivalence fails as soon as a more meaningful likelihood concept, which assigns probabilities to all trees that are displayed by a network, is applied. However, it is easily seen that, if all such trees have the same probability, the rankings suggested by softwired parsimony and weighted softwired likelihood are identical and, thus, softwired MP and weighted softwired ML choose the same optimal networks. On the other hand, if one wanted to employ a (more biologically plausible) non-uniform distribution on the trees displayed by a phylogenetic network, we conjecture that softwired parsimony and weighted softwired likelihood are equivalent if one changes the definition of softwired parsimony in a way that assigns a suitable scaling factor to each displayed tree. For future research, it will be interesting to prove this conjecture and to analyze how other definitions of likelihood on networks relate to parsimony. \\

\noindent {\bf Acknowledgements.}  We wish to thank Steven Kelk for helpful discussions on the topic. The first, third, and fourth author thank the New Zealand Marsden Fund for their financial support. Part of this work was done while the first author was a student at the University of Canterbury, New Zealand.


\begin{thebibliography}{00}


\bibitem[Bapteste et al.(2013)]{bapteste13}
E. Bapteste, L. van Iersel, A. Janke, S. Kelchner, S. Kelk, J. O. McInerney, D. A. Morrison, L. Nakhleh, M. Steel, L. Stougie, and J. Whitfield (2013). Networks: expanding evolutionary thinking. Trends in Genetics, 29, 439--441.

\bibitem[Barry and Hartigan(1987)]{barryhartigan}
D. Barry and J. Hartigan (1987). Statistical analysis of hominoid molecular evolution. Statistical Science, 2:191--207.

\bibitem[Felsenstein(1981)]{felsenstein1981}
J. Felsenstein (1981). Evolution trees from DNA sequences: a maximum likelihood approach. Journal of Molecular Evolution, 17, 368--376.

\bibitem[Felsenstein(2004)]{felsenstein04}
J. Felsenstein (2004). Inferring phylogenies. Sinauer Associates.

\bibitem[Fischer et al.(2015)]{fischer}
M. Fischer, L. van Iersel, S. Kelk, and C. Scornavacca (2015). On computing the maximum parsimony score of a phylogenetic network. SIAM Journal of Discrete Mathematics, 29, 559--585.

\bibitem[Fitch(1971)]{fitch71}
W. M. Fitch (1971). Toward defining the course of evolution: minimum change for a specific tree topology. Systematic Biology, 20, 406--416.

\bibitem[Foulds and Graham(1982)]{foulds82}
L. R. Foulds and R. L. Graham (1982). The Steiner Problem in phylogeny is NP-complete. Advances in Applied Mathematics, 3, 43--49.

\bibitem[Francis and Steel(2015)]{francis15}
A. R. Francis and M. Steel (2015). Which phylogenetic networks are merely trees with additional arcs? Systematic Biology, 64, 768--777. 

\bibitem[Hartigan(1973)]{hartigan73}
J. A. Hartigan (1973). Minimum mutation fits to a given tree. Biometrics, 29, 53--65.

\bibitem[Jin et al.(2006)]{jin06}
G. Jin, L. Nakhleh, S. Snir, and T. Tuller. (2006). Efficient parsimony-based methods for phylogenetic network reconstruction. Bioinformatics, 23, e123--e128.

\bibitem[Jin et al.(2007)]{jin07}
G. Jin, L. Nakhleh, S. Snir, and T. Tuller. (2007). Inferring phylogenetic networks by the maximum parsimony criterion: a case study. Molecular Biology and Evolution, 24, 324--337.

\bibitem[Jin et al.(2009)]{jin09}
G. Jin, L. Nakhleh, S. Snir, and T. Tuller (2009). Parsimony score of phylogenetic networks: hardness results and a linear-time heuristic. IEEE Transactions on Computational Biology and Bioinformatics, 6, 495--505.

\bibitem[Jukes and Cantor(1969)]{JC}
T. Jukes and C. Cantor (1969). Evolution of protein molecules. Mammalian Protein Metabolism, 3, 21--132.

\bibitem[Kannan and Wheeler(2014)]{kannan14}
L. Kannan and W. C. Wheeler (2014). Exactly computing the parsimony scores on phylogenetic networks using dynamic programming. Journal of Computational Biology, 21, 1--17.

\bibitem[Kannan and Wheeler(2012)]{kannan12}
L. Kannan and W. C. Wheeler (2012). Maximum parsimony on phylogenetic networks. Algorithms in Molecular Biology, 7:9.

\bibitem[McDiarmid et al.(2015)]{mcdiarmid14}
 C. McDiarmid, C. Semple, and D. Welsh (2015). Counting phylogenetic networks. Annals of Combinatorics, 19, 205--224.

\bibitem[Morrison(2011)]{morrison11} 
D. A. Morrison (2011). Introduction to phylogenetic networks. RJR Productions.


\bibitem[Nakhleh et al.(2005)]{nakhleh05}
L. Nakhleh, G. Jin, F. Zhao, and J. Mellor-Crummey (2005). Reconstructing phylogenetic networks using maximum parsimony. In: Proceedings of the 2005 IEEE Computational Systems Bioinformatics Conference (CSB2005), pp. 93--102. 

\bibitem[Nakhleh(2011)]{NA}
L. Nakhleh (2011). Evolutionary phylogenetic networks: models and issues. In: Problem solving handbook in computational biology and bioinformatics, pp. 125--158.

\bibitem[Neyman(1971)]{N}
J. Neyman (1971). Molecular studies of evolution: a source of novel statistical problems. In: Statistical decision theory and related topics, pp. 1--27.

\bibitem[Nguyen et al.(2007)]{nguyen07}
C. T. Nguyen, N. B. Nguyen, W.-K. Sung, and L. Zhang (2007). Reconstructing recombination network from sequence data: the small parsimony problem.  IEEE Transactions on Computational Biology and Bioinformatics, 4, 394--402.

\bibitem[Sankoff(1975)]{sankoff75}
D. Sankoff (1975). Minimal mutation trees of sequences. SIAM Journal of Applied Mathematics, 28, 35--42.

\bibitem[Semple and Steel(2003)]{semple03}
C. Semple and M. Steel (2003). Phylogenetics. Oxford University Press.

\bibitem[Steel and Penny(2000)]{steelpenny}
M. Steel and D. Penny (2000). Parsimony, likelihood, and the role of models in molecular phylogenetics. Molecular Biology and Evolution, 17:839--850.

\bibitem[Tuffley and Steel(1997)]{TS}
C. Tuffley and M. Steel (1997). Links between maximum likelihood and maximum parsimony under a simple model of site substitution. Bulletin of Mathematical Biology, 59:581--607.

\bibitem[Wheeler(2015)]{wheeler15}
W. C. Wheeler (2015). Phylogenetic network analysis as a parsimony optimization problem. BMC Bioinformatics, 16:296.


\end{thebibliography}


\end{document}